\documentclass[letterpaper,11pt]{article}
\usepackage[letterpaper,margin=1in]{geometry}
\usepackage{booktabs} % For formal tables
\usepackage[ruled,vlined,linesnumbered]{algorithm2e} % For algorithms
\usepackage{amsmath, amssymb, amsthm}
\usepackage{dsfont}
\usepackage[backref=page]{hyperref}
\usepackage[capitalize]{cleveref}
\usepackage{xcolor}
\usepackage[shortlabels,inline]{enumitem}

\SetAlFnt{\small}
\SetAlCapFnt{\small}
\SetAlCapNameFnt{\small}
\SetAlCapHSkip{0pt}
\SetKwInOut{Input}{Input}
\SetKwInOut{Output}{Output}

\newtheorem{theorem}{Theorem}[section]
\newtheorem{lemma}[theorem]{Lemma}

{\theoremstyle{definition}

}

\newenvironment{proofof}[1]{\begin{proof}[Proof of {#1}]}{\end{proof}}
\newcommand{\Ex}{\mathbb{E}}

\newcommand{\OPT}{\operatorname{OPT}}
\newcommand{\argmax}{\operatorname{argmax}}

\newcommand{\LPmax}[1][\kappa,S]{\ensuremath{\bigl(\operatorname{LP}{(#1)}\bigr)}}
\newcommand{\Vmax}{\operatorname{OPT_{LP}}}

\newcommand{\sm}{\ensuremath{\setminus}}

\newcommand{\sse}{\ensuremath{\subseteq}}

\newcommand{\R}{\ensuremath{\mathbb R}}

\newcommand{\Mc}{\ensuremath{\mathcal M}}

\newcommand{\Sc}{\ensuremath{\mathcal S}}

\newcommand{\e}{\ensuremath{\epsilon}}
\newcommand{\gm}{\ensuremath{\gamma}}
\newcommand{\Gm}{\ensuremath{\Gamma}}

\newcommand{\kp}{\ensuremath{\kappa}}
\newcommand{\al}{\ensuremath{\alpha}}

\newcommand{\Dt}{\ensuremath{\Delta}}

\newcommand{\dlp}[1][\kappa,S]{\ensuremath{\bigl(\operatorname{DLP}({#1})\bigr)}}
\newcommand{\xos}{\ensuremath{\mathsf{xos}}}
\newcommand{\demd}{\ensuremath{\mathsf{demd}}}
\newcommand{\tc}{\ensuremath{\tilde c}}
\newcommand{\bA}{\ensuremath{\overline A}}
\newcommand{\bR}{\ensuremath{\overline R}}

\title{An $O(\log \log n)$-approximate budget feasible mechanism for subadditive
  valuations\thanks{An extended abstract is to appear in the Proceedings of the 26th ACM
    Conference on Economics and Computation (EC) 2025.}}

\author{
    Rian Neogi\thanks{\texttt{\{rneogi,kpashkovich,cswamy\}@uwaterloo.ca}. 
    Dept. of Combinatorics and Optimization, Univ. Waterloo, Waterloo, ON N2L 3G1. 
    Supported in part by the NSERC Discovery grants of K. Pashkovich and C. Swamy.}
\and
\addtocounter{footnote}{-1} 
Kanstantsin Pashkovich\footnotemark 
\and 
\addtocounter{footnote}{-1} 
Chaitanya Swamy\footnotemark
}

\date{}

\begin{document}

% Choose a citation style by commenting/uncommenting the appropriate line:
%\setcitestyle{acmnumeric}
% \setcitestyle{authoryear}

\maketitle

% Abstract. Note that this must come before \maketitle.
\begin{abstract}
In \emph{budget-feasible mechanism design}, there is a set of items $U$, each owned by a distinct seller. The seller of item $e$ incurs a private cost $\overline{c}_e$ for supplying her item. A buyer wishes to procure a set of items from the sellers of maximum value, where the value of a set $S\subseteq U$ of items is given by a valuation function $v:2^U\to \mathbb{R}_+$. The buyer has a budget of $B \in \mathbb{R}_+$ for the total payments made to the sellers. We wish to design a mechanism that is \emph{truthful}, that is, sellers are incentivized to report their true costs,  \emph{budget-feasible}, that is, the sum of the payments made to the sellers is at most the budget $B$, and that outputs a set whose value is large compared to $\text{OPT}:=\max\{v(S):\overline{c}(S)\le B,S\subseteq U\}$.

Budget-feasible mechanism design has been extensively studied, with the
literature focusing on (classes of) \emph{subadditive} (or complement-free) valuation
functions, and various polytime, budget-feasible mechanisms, achieving constant-factor
approximation to $\text{OPT}$, have been devised for the special cases of additive,
submodular, and XOS (or fractionally subadditive) valuations. However, for general
subadditive valuations, the best-known approximation factor achievable by a polytime
budget-feasible mechanism (given access to demand oracles) was only $O(\log n / \log \log
n)$, where $n=|U|$ is the number of items.  

We improve this state-of-the-art significantly by designing a randomized budget-feasible
mechanism for subadditive valuations that \emph{achieves a substantially-improved
  approximation factor of $O(\log\log n)$ and runs in polynomial time, given access to
  demand oracles.} 

Our chief technical contribution is to show that, given any set $S\subseteq U$, one can
construct in polynomial time a distribution $\mathcal{D}$ over posted-payment vectors $d
\in\mathbb{R}_+^S$ satisfying $d(S)\le B$ (where $d(S)=\sum_{e \in S}d_e$) such that,
$\mathbb{E}_{d \sim \mathcal{D}}[v(\{e \in S: c_e\le d_e\} )] \ge \frac{1}{O(\log \log
  n)}\cdot v(S)$ for every cost-vector $c \in \mathbb{R}^S_+$ satisfying $c(S)\le B/O(\log
\log n)$. Using this distribution, we show how to construct a budget feasible mechanism
for subadditive valuations that achieves an approximation factor of $O(\log \log n)$. 
\end{abstract}

\section{Introduction}

In \emph{budget-feasible mechanism design}, there is a ground set $U$ of $n$ elements or
items, and a valuation function $v : 2^U\to \mathbb{R}_+$, where $v(S)$ specifies the
value of a subset $S \subseteq U$. We are also given a budget $B \in \mathbb{R}_+$. 
The mechanism designer (or buyer) seeks to procure a set of elements of maximum value.
Each element $e$ is held by a strategic seller who incurs a private cost
$\overline{c}_e\ge 0$ for supplying item $e$; we will often identify sellers with the
items they hold. 
The central difficulty here is that the private costs $\{\overline{c}_e\}$ are known only to
the sellers, and not to the mechanism designer. 
Each seller $e$ reports a cost $c_e$ to the mechanism, which may not be equal to her true
cost.
In order to incentivize sellers to reveal their true costs, the mechanism designer makes payments to the sellers.
The utility of a seller is equal to the payment it receives minus the cost incurred by potentially supplying the item.
Our goal is to design a \emph{truthful} mechanism, where each seller maximizes their utility by reporting her true cost. Additionally, the mechanism designer also has a budget $B$ and we need to ensure that the sum of the payments given out to the sellers do not exceed this budget. 

We want to design a mechanism that outputs a set that is a good approximation to the best set one can get subject to the budget constraint. 
That is to say, we compare the value of our solution to the benchmark $\OPT:=\max\{v(S):c(S)\le B\}$. 
Note that for a truthful mechanism, it does not matter whether we consider the true costs or the reported costs for the budget constraint.
For general valuation functions as input, the work of \cite{Singer10} shows that no non-trivial approximation is possible for any budget feasible mechanism.
Hence work has been focussed on special classes of valuation functions, such as,
typically (in increasing order of generality) additive, submodular, XOS, or subadditive valuations.

For subadditive valuations, ~\cite{Bei2017WorstCaseMD} made an important
contribution showing that the problem admits a randomized 
$O(\log n / \log \log n)$-approximate budget feasible mechanism. (Recall that $n=|U|$).
They also provide an $O(1)$-factor approximate mechanism in the Bayesian setting, where costs are sampled from a prior distribution that is known to the mechanism designer.
They use Yao's minimax principle to show that their $O(1)$-approximate Bayesian mechanism
implies the \emph{existence} of an $O(1)$-factor approximate mechanism in the prior-free
setting. 
This work was later improved in two ways: \cite{BalkanskiGGST22} achieved a
\emph{deterministic} $O(\log n / \log \log n)$ approximation, 
and~\cite{NeogiPS24} provide a \emph{constructive} $O(1)$-factor approximate mechanism in
the prior-free setting, that runs in exponential time.

\subsection{Our contributions}
{\em We design a budget-feasible mechanism for subadditive valuations that achieves an
$O(\log \log n)$-approximation, 
and runs in polynomial time given access to demand oracles}.
This is a substantial improvement over the previous-best approximation factor of 
$O(\log n / \log \log n)$ achievable in polytime (also given demand
oracle)~\cite{Bei2017WorstCaseMD,BalkanskiGGST22}.  

\begin{theorem}\label{thm:main_result} \label{mainthm}
There is a truthful budget feasible mechanism for subadditive valuations that achieves an
approximation factor of $O(\log \log n)$, and runs in polynomial time given a demand oracle
for the valuation $v$.
\end{theorem}

Our mechanism is randomized: it satisfies truthfulness and the budget constraint with
probability 1, and obtains expected value at least $\OPT/O(\log\log n)$.
To state our key technical contribution, 
it will be useful to discuss some snippets from prior work on budget-feasible
mechanism design for subadditive valuations.
An important insight to emerge from the work of~\cite{Bei2017WorstCaseMD,NeogiPS24} is the
following structural property of subadditive functions, stated explicitly in~\cite{NeogiPS24}.
(This is quite implicit in~\cite{Bei2017WorstCaseMD} hiding behind their use of Yao's
principle.)  

\begin{theorem}[\cite{NeogiPS24}] \label{thm:existence_of_dist}
Let $v : 2^U \to \mathbb{R}_+$ be a subadditive function and $S^*\sse U$. Let 
$K \subseteq \mathbb{R}^{S^*}_+$ be a finite set. 
There exists a distribution $\mathcal{D}^*$ over vectors in $K$ such that, for any 
$c\in K$, we have $\Ex_{d \sim \mathcal{D}^*}[v(\{e \in S^*:c_e\le d_e\} )] \ge \frac{1}{2}v(S^*)$. 
\end{theorem}

For completeness, we include the proof of the above theorem in the appendix. 
\cite{NeogiPS24} show that this leads to an $O(1)$-approximate budget-feasible mechanism 
by computing a suitable set $S^*$ of large value, and 
taking $K$ to a discretized version of the simplex 
$\Dt^{S^*}(B):= \{x \in \mathbb{R}^{S^*}_+: \sum_{e\in S^*}x_e\le B\}$: 
a vector $d$ sampled from the distribution $\mathcal{D}^*$ yields 
take-it-or-leave-it posted payments for the sellers (that sum to at most $B$), 
so that if we select seller $e$ only if $c_e\le d_e$ we obtain expected value at least 
$\frac{1}{2}v(S^*)$. 
However, computing the distribution in \cref{thm:existence_of_dist} entails solving an
exponentially large LP, and so, this only yields an {\em exponential-time mechanism}. 

Our {\em chief technical contribution} is to prove the following relaxed version of
\cref{thm:existence_of_dist}, which shows that, for the simplex $\Dt^{S^*}(B)$, a
distribution satisfying a certain weaker guarantee {\em can be computed in polytime}. For 
a vector $\al\in \mathbb{R}^U$ an $S\sse U$, let $\al(S)$ denote $\sum_{e\in S}\al_e$.

\begin{theorem}[Informal version of Theorem~\ref{thm:good_dist}]\label{thm:good_dist_informal}
Let $v:2^U \to \mathbb{R}_+^U$ be a subadditive function, $S^*\sse U$, and  $B \in \mathbb{R}_+$.
Using demand oracles, one can compute in polynomial time a distribution $\mathcal{D}$ over
vectors in $\Dt^{S^*}(B)$,
such that for every $c\in \mathbb{R}^{S^*}_+$ satisfying 
$c(S^*) \le \frac{B}{O(\log \log n)}$, we have the guarantee
$\Ex_{d \sim \mathcal{D}}[v(\{e \in S^*:c_e\le d_e\} )]\ge \frac{v(S^*)}{O(\log \log n)}$.%
\footnote{The set $K$ in Theorem~\ref{thm:existence_of_dist} is finite, but here the
guarantee holds for all $c$ satisfying the stated condition.}
\end{theorem}

To utilize \cref{thm:good_dist_informal} to obtain our $O(\log\log n)$-approximate
mechanism, our plan, roughly speaking, is to use \cref{thm:good_dist_informal} in the same
fashion as \cref{thm:existence_of_dist}:
that is,
we use a vector $d$ sampled from the distribution $\mathcal{D}$ in
\cref{thm:good_dist_informal} as posted payments for the sellers. 
But implementing this plan takes some work in order to deal with the case where the
cost-vector $c$ is not handled by Theorem~\ref{thm:good_dist_informal} (i.e., we have 
$c(S^*)>B/O(\log\log n)$). 
We cannot simply use the mechanisms of~\cite{NeogiPS24,Bei2017WorstCaseMD}, 
replacing \cref{thm:existence_of_dist} with \cref{thm:good_dist_informal}, because
these mechanisms use oracles, such as knapsack-cover oracle~\cite{NeogiPS24}, approximate
XOS oracle~\cite{Bei2017WorstCaseMD}, that are different from (and somewhat more-powerful
than) a demand oracle.  
(In~\cite{NeogiPS24}, the set $S^*$ is obtained using an oracle they call a
knapsack-cover oracle; whereas in~\cite{Bei2017WorstCaseMD}, 
the case where $c(S^*)$ is large is dealt with by assuming that one has an XOS oracle for
an XOS-approximation of $v$.)
We assume only demand-oracle access in Theorem~\ref{thm:main_result}, so  
we need to come up with some novel ideas to deal with the case where $c(S^*)$ is large. 
Combining these with \cref{thm:good_dist_informal} leads to \cref{thm:main_result}.

\subsection{Technical overview}\label{sec:tech_overview}
\cref{thm:good_dist_informal} constitutes the technical core 
of our mechanism, and we begin by discussing the key ideas that go toward proving this result.

Fix a subadditive valuation $v:2^U\mapsto \mathbb{R}_+$.
For $B \in \mathbb{R}_+$ and $S \subseteq U$, let 
$\Delta^S(B) := \{x \in \mathbb{R}^S_+ : \sum_{e \in S}x_e \le B\}$.
We use $\Delta (B)$ as a shorthand for $\Delta ^U(B)$.
Fix $\gamma  \in (0,1]$. We focus on the case $S^*=U$ in
Theorem~\ref{thm:good_dist_informal}, which will allow us to convey 
the main ideas. 

The statements of \cref{thm:existence_of_dist} and
\cref{thm:good_dist_informal} can be viewed as statements about optimal strategies in a
two-player zero-sum 
item bidding game. 
We refer to the two players as the $d$-player and the $c$-player.
The $d$-player chooses a bid vector $d \in \Delta (B)$. 
For each $e\in U$, we interpret $d_e$ as the $d$-player's bid for the item $e$. The
restriction $d\in \Delta (B)$ ensures that the sum of the bids of the $d$-player
is within the budget $B$. 
The $c$-player chooses another vector $c \in \Delta (\gamma B)$, representing the
$c$-player's bids\footnote{The use of notation $c$ here is suggestive. We will later take
  this vector to be the vector of reported costs $\{c_e\}_{e \in U}$ from the
  budget-feasible mechanism design problem.}.  
The payoff of the game for the $d$-player is $v(\{e \in U: c_e\le d_e\} )$, that is, it is
the value of the set of items where the $d$-player's bid is higher than the $c$-player's
bid. The goal of the $d$-player is to maximize this payoff, and goal of the $c$-player is
to minimize this payoff. 

By von Neumann's minimax theorem,
this zero-sum game has a Nash Equilibrium over mixed-strategies for the $d$-player and
$c$-player, and the payoff to the $d$-player under a mixed Nash equilibrium is equal to
the max-min value of the game.
\footnote{von Neumann's minimax theorem requires a finite space of pure strategies, and
   so, more precisely, we need to work with a discretization $K$ of the strategy space
   $\Delta(B)$; accordingly, in \eqref{eq:NE}, the $\min_{c\in\Dt(B)}$ should really be
   $\min_{c\in K}$. 
  Under the assumption that the sellers use some polynomial number $s$ of bits to describe
  their reported cost, such a discretization is always possible since we may limit ourselves to
  $2^s$ possibilities for each entry of $d$ and $c$.}
Note that a mixed-strategy for the $d$-player is a distribution over vectors in 
$\Delta(B)$, and for the $c$-player, is a distribution over vectors in $\Delta (\gamma B)$. 
The max-min value corresponds to the $d$-player going first, followed by the $c$-player, so
we may assume that the $c$-player plays pure strategies only. 
Taking $K$ to be (more precisely, a discretization of) $\Delta (B)$, 
\cref{thm:existence_of_dist} shows that, for any $\gamma \le 1$, 
this max-min value is at least $\frac{1}{2}v(U)$. 
That is, 
\begin{equation}\label{eq:NE}
	\max_{\substack{\text{distribution } \mathcal{D} \\ \text{ over vectors in } \Delta (B)}} \min_{c \in \Delta (B)} \Ex_{d \sim \mathcal{D}} [v(\{e \in U:c_e\le d_e\} )] \ge \frac{1}{2}v(U).
\end{equation}

\cref{thm:good_dist_informal} can also be interpreted as a statement about finding
mixed-strategies for the item bidding game.  
It says that, when $\gamma = O(\frac{1}{\log \log n})$, 
{\em we can compute in polynomial time}, a mixed strategy for the $d$-player (that is, a
distribution over vectors in 
$\Delta (B)$) that achieves a payoff that is a $O(\log \log n)$-approximation to the
optimal payoff. 
That is,
\begin{equation}\label{eq:loglogn_dist}
	\min_{c \in \Delta (B/O(\log \log n))} \Ex_{d \sim \mathcal{D}} [v(\{e \in U:c_e\le d_e\} )] \ge \frac{v(U)}{O(\log \log n)}
\end{equation}

With this two-player-game view in mind,  
we next describe how to efficiently compute a good mixed-strategy for the $d$-player in
the item bidding game, which  
leads to the proof of \cref{thm:good_dist_informal}.

\paragraph{Finding good strategies.}
As mentioned earlier, computing the distribution in \cref{thm:existence_of_dist} entails
solving an exponentially-large LP, with variables and constraints for the players' pure
strategies, 
which we cannot hope to do in polynomial time even with demand oracles. 
Nonetheless, let us attempt to find such a distribution even if the guarantee on its
payoff is suboptimal. 
When $\gm=1$ 
both players 
pick a vector in $\Delta (B)$.  
In this setting, one can show that any pure strategy of the $d$-player fails to give 
payoff better than $\max_{e\in U}v(e)$ 
in the worst case (see 
\cref{lem:no_pure_strategy}). 
Thus, the $d$-player must resort to using mixed strategies in order to get a decent payoff.

How should the $d$-player randomize?
To gain intuition, suppose that there 
we have $k$ disjoint sets $S_1,\ldots,S_k\sse U$ with the property that each set $S_i$ has
large value, that is, $v(S_i)\ge \alpha v(U)$ for all $i=1,\ldots,k$, for some fixed
$\alpha \in (0,1)$.  
The $d$-player can utilize this partition to gain an advantage over the $c$-player.
The $d$-player selects a set $S$ from $S_1,\ldots,S_k$ uniformly at random and focuses
their entire budget $B$ on elements in $S$, giving bids of zero to elements in 
$U\setminus S$. The $c$-player does not know the random choice of the $d$-player.
No matter what bids are chosen by $c$-player, we will always have that $\Ex[c(S)] \leq \frac{B}{k}$. 
This is because, for any item $e$, the probability that $e\in S$ is at most $\frac{1}{k}$.

Herein lies the advantage of using randomization: the $d$-player gets to use their entire budget
of $B$ on the set of items in $S$, whereas the $c$-player has a budget of only
$\frac{B}{k}$ (in expectation) to use for items in $S$. 
Thus, the $d$-player can utilize the sets $S_1,\ldots,S_k$ to effectively boost his own
budget by a factor of $k$. 
As we will show, the diminished bidding power for the $c$-player over the items in $S$ will
make it possible for the $d$-player to find a \emph{pure strategy} over $S$ that achieves
a large payoff compared to $v(S)$. This is in contrast to the situation where both the
$d$-player and the $c$-player have the same budget $B$, for which no pure strategy
provides a good payoff for the $d$-player (\cref{lem:no_pure_strategy}).  
Since $v(S) \ge \alpha v(U)$, a strategy that attains a large payoff compared to $v(S)$
also attains a large payoff compared to $v(U)$, albeit with an $\alpha $-factor loss. 

This is, in essence, the idea underlying how we compute the distribution $\mathcal{D}$ in
\cref{thm:good_dist_informal}.  
However, there are a few things left unspecified and/or are oversimplified in the above exposition. 
(1) 
For what values of $\alpha $ and $k$ do there exist such disjoint sets 
$S_1,\ldots,S_k$ of $U$ such that $v(S_i)\ge \alpha v(U)$ for all $i=1,\ldots,k$, and how
can we find these sets in polynomial time?
(2) Once a set $S\in\{S_1,\ldots,S_k\}$ is sampled, how can the $d$-player find a pure
strategy over the items in $S$ that guarantees good payoff? 
(3) 
We have $\Ex[c(S)]\le \frac{B}{k}$, 
but it need not be true that the sampled set satisfies $c(S)\le \frac{B}{k}$ always.

Dealing with (3) is straightforward. We simply use Markov's inequality to bound $c(S)$. 
For any $\epsilon >0$, Markov's inequality guarantees that $c(S) \le \frac{B}{\epsilon k}$
with probability at least $ 1-\epsilon $. 
This has the downside of boosting the effective budget of the $c$-player by a factor of
$\frac{1}{\epsilon }$. 
We remedy this by asserting that the $c$-player pick a vector $c \in \mathbb{R}^U_+$ such
that $c(U) \le \epsilon B$ in the first place, that is, we assume that the $c$-player is
already restricted to a limited budget of $\epsilon B$. 
With this assumption, using Markov's inequality, we get that $c(S) \le \frac{B}{k}$ with
probability at least $1-\epsilon $. 
We end up using $\epsilon = \frac{1}{O(\log \log n)}$, which is why we require that $c(U)
\le  \frac{B}{O(\log \log n)}$ in the statement of \cref{thm:good_dist_informal}. 

\medskip
Dealing with issues (1) and (2) is more involved.
To address (1), we observe that we do not really need disjoint sets $S_1,\ldots,S_k$ with
$v(S_i)\ge \alpha v(U)$ for all $i=1,\ldots,k$. 
Sampling a set $S$ uniformly from $S_1,\ldots,S_k$ induces a distribution over sets with
the properties (i) $\Ex_S[v(S)]\ge \alpha v(U)$ and (ii) $\Pr[e \in S] \le \frac{1}{k}$. 
{\em We really only need a distribution $\mathcal{S}$ over subsets $S \subseteq U$
satisfying (i) and (ii) for our argument.} 
Property (i) ensures that the $d$-player retains a large fraction of the total value when
sampling a set $S$ from $\mathcal{S}$, and 
property (ii) ensures that $\Ex[c(S)] \le \frac{B}{k}$ when sampling $S$ from $\mathcal{S}$. 
This naturally leads us to the following LP, which is parameterized by a subset 
$S \subseteq U$ and value $\kappa \in (0,1)$, which finds the distribution of largest
expected value with marginals at most $\kappa$. 
\begin{alignat*}{2}
\text{maximize} & \quad & \sum_{T \subseteq S}v(T)& x_T \tag*{$\LPmax$} \label{distrlp} \\
\text{subject to} & \quad & \sum_{T : e \in T}x_T  &\le \kappa \qquad \forall e \in S \\
&& \sum_{T \subseteq S}x_T & \le 1 \\
&& x  & \ge 0.
\end{alignat*}
A feasible solution $x$ to \ref{distrlp} yields
a distribution over subsets $T$ of $S$, where 
we sample a set $T\neq \emptyset$ with probability $x_T$, and 
sample $\emptyset$ with probability $1-\sum_{T' \subset S}x_{T'}$. 
(which is non-negative due to the second constraint).  
The constraint $\sum_{T : e \in T}x_T \le \kappa $ for all $e \in S$ ensures that
$\Pr_{T}[e \in T] \le \kappa $. 
Moreover, the objective function $\sum_{T \subseteq S}v(T)x_T$ is simply $\Ex_{T}[v(T)]$.

It is not hard to show that we can solve \ref{distrlp} in polynomial time using demand
oracles, since the separation problem for the dual LP (see \ref{distrdual}) amounts to
precisely a demand-oracle computation.
Let $\Vmax(\kappa ,S)$ denote the optimal value of $\LPmax$.
The distribution corresponding to this optimal solution now replaces the partition $S_1,\ldots,S_k$.
The $d$-player now samples a set $S$ from this distribution and restricts their bids to elements in $S$.
In the special case when there exists a partition $S_1,\ldots,S_k$ satisfying $v(S_i)\ge \alpha v(U)$ for all $i=1,\ldots,k$, the uniform distribution over $S_1,\ldots,S_k$ is a feasible solution to $\LPmax[{\kappa,U}]$, showing that $\Vmax(\kappa ,U)\ge \alpha v(U)$.
However, the LP optimizes over a much richer family of distributions, and so the value of the LP can be high even when such a partition does not exist.

\medskip
Finally, to address (2), we show that the \emph{optimal solution to the dual of
\ref{distrlp}} (see \ref{distrdual}) can be used to design a pure strategy (that is, a
vector of bids) for the $d$-player. This is captured by the following lemma.

\begin{lemma}[Informal version of Lemma~\ref{lem:thresholds_simple}]\label{lem:thresholds_informal}
Fix $B \in \mathbb{R}_+$, $\kappa \in (0,1)$ and $S \subseteq U$.
Using demand oracles, we can compute in polynomial time, a vector 
$d \in \mathbb{R}^S_+$ such that $d(S) = B$, and 
$v(\{e \in S : c_e\le d_e\} ) \ge v(S)-\Vmax(\kappa ,S)$ holds for all 
$c\in\Dt^S(\kp B)$.
\end{lemma}

In other words, the $d$-player can play a pure strategy over the items in $S$ that achieves good payoff as long as $\Vmax(\kappa ,S)$ is small.

This discussion leads to the following modified approach for the $d$-player.
\begin{enumerate}
	\item First, sample a set $S \subseteq U$ from the distribution given by the optimal solution of $\LPmax[{\kappa,U}]$ for some $\kappa $ to be decided later.
	\item Then, choose a good pure strategy for the $d$-player over the items in $S$ using \cref{lem:thresholds_informal}. Extend it to a pure strategy for all the items $U$ by setting zeroes for entries corresponding to elements in $U\setminus S$.
\end{enumerate}
The above two step process leads to a mixed strategy for the $d$-player, where the randomness comes from the random choice of $S$.

The only thing remaining is to decide what value of $\kappa$ to pick.
This is a balancing act. If $\kappa $ is too small, 
it could be that $\Vmax(\kappa ,U)$ is too small, which means that the $d$-player has no
hope of achieving a good payoff, even if he wins all the items in $S$. 
If $\kappa $ is too large, then \cref{lem:thresholds_informal} might fail to guarantee a good payoff (to see this, consider the extreme case of $\kappa =1$, in which case $\Vmax(\kappa ,S)=v(S)$ and so \cref{lem:thresholds_informal} provides no guarantee on the payoff).

For a fixed subset $S \subseteq U$, \cref{lem:thresholds_informal} guarantees a payoff of at least $v(S)-\Vmax(\kappa ,S)$.
When this set $S$ is sampled from the distribution $\mathcal{S}$ coming from the optimal solution of $\LPmax[{\kappa,U}]$, the expected payoff is at least $\Ex_{S \sim \mathcal{S}} [v(S)-\Vmax(\kappa ,S)]$.
The first term $\Ex_{S\sim \mathcal{S}}[v(S)]$ is simply $\Vmax(\kappa ,U)$.
It turns out that the second term $\Ex_{S\sim \mathcal{S}}[\Vmax(\kappa ,S)]$ can be upper bounded by $\Vmax(\kappa^2,U)$.
Thus, the expected payoff for the $d$-player using this strategy is at least $\Vmax(\kappa ,U)-\Vmax(\kappa ^2,U)$.
A result of \cite{Dtting2020AnOL} shows that there exists $\kappa \in (0,1)$ for which
$\Vmax(\kappa ,U)-\Vmax(\kappa ^2,U)$ is at least $\frac{1}{O(\log \log n)}\cdot v(U)$. 
This is the value of $\kappa $ that we will use for step (1).
Combining everything, we obtain a distribution $\mathcal{D}$ over vectors in $\Delta (B)$, for which $\Ex[v(\{e \in U:c_e\le d_e\} )]\ge \frac{1}{O(\log \log n)}v(U)$ for all $c \in \Delta (\frac{B}{O(\log \log n)})$.

\paragraph{Constructing the mechanism.}
To complement \cref{thm:good_dist_informal}, we need to 
show how one can use the distribution (or equivalently, mixed-strategy for the $d$-player) 
provided by \cref{thm:good_dist_informal} to design a budget-feasible mechanism achieving
a good approximation. 
Our mechanism (see Algorithm~\ref{bfmech-alg} in Section~\ref{bfmech})
builds upon the framework used by \cite{Bei2017WorstCaseMD} (also later used
by \cite{NeogiPS24}). 
We randomly partition the set of sellers into two groups $U_1,U_2$, get an estimate of
$\OPT$ from the sellers in $U_1$ who are then discarded, 
and work over the sellers in $U_2$. 
The set $S^*$ we use in \cref{thm:good_dist_informal} is obtained by computing 
a demand set over $U_2$ with appropriately chosen prices. 
This ensures that $v(S^*)=\Omega(\OPT)$.
As discussed earlier, we sample a posted-payment vector $d$ from the distribution $\mathcal{D}$ in
\cref{thm:good_dist_informal}, which  
yields the set of elements $R = \{e \in S^*:c_e\le d_e\}$. 
The value $d_e$ can be thought of as the \emph{threshold} for seller $e$, in the sense
that $e$ will not be picked into the solution if it declares a cost higher than
$d_e$. So we often refer to the vector $d$ as a \emph{threshold vector} for the elements in 
$S^*$. 
The resulting mechanism is truthful (due to Myerson's characterization of
truthfulness, \cref{lem:myersons}) and has payments bounded by $d_e$, and since  
$d(U)\leq B$, we obtain budget-feasibility. 
If $c(S^*)\le \frac{B}{O(\log \log n)}$, the approximation factor follows from the
guarantee in \cref{thm:good_dist_informal}. But we also need to deal with the case where
this does not hold.
This presents some technical challenges, and here, we need to proceed significantly
differently from the mechanisms of~\cite{NeogiPS24,Bei2017WorstCaseMD}. 
\cite{Bei2017WorstCaseMD} show in the analysis of their mechanism for subadditive
valuations in the Bayesian setting,
that when $c(S^*) > \frac{B}{\alpha }$, 
the subadditive valuation $v$ is actually ``$\alpha $-close to an XOS function'', and 
so one can use a mechanism for XOS valuations in this case.
However, this requires one to have a suitable oracle for this
XOS-approximation $v^\xos$ of $v$: either an XOS-oracle for $v^\xos$, or a demand oracle
for $v^\xos$, as required by the mechanisms for XOS valuations
in~\cite{Bei2017WorstCaseMD} and~\cite{NeogiPS24} respectively. 
We only have a demand oracle for $v$, so this presents a significant challenge. 
In the mechanism of~\cite{NeogiPS24} for subadditive valuations, $S^*$ is obtained
using what they call a knapsack-cover oracle for $v$, which ensures that
$c(S^*)\leq B$, so they can always resort to \cref{thm:existence_of_dist}. 
In our case, we do not have a knapsack-oracle for $v$, and we are using
\cref{thm:good_dist_informal}, so even if we 
had $c(S^*)\leq B$, we still need to deal with the case where $c(S^*)>B/O(\log\log n)$. 

We proceed instead as follows. 
We show that when $v$ is $\alpha$-close to an XOS function---which happens when
$c(S^*)>B/\al$---that we can solve another LP using a demand oracle to obtain another
suitable threshold-vector $d'$ (see steps~\ref{xosthresh}, \ref{threshprune} of
Algorithm~\ref{bfmech-alg}).  
We utilize a pruning operation to obtain a set $R'$ such that: 
$c_e \le d'_e$ for all $e\in R'$, and $d'(R')\leq B$.
We show that these two properties imply that $v(R')=\Omega(\OPT)$, 
and that it can be procured with payments at most the budget.

\subsection{Related work}

Budget-feasible mechanism design was introduced by Singer~\cite{Singer13}, who also gave $O(1)$-factor randomized mechanisms for submodular valuations.
This was followed by a sequence of work \cite{BalkanskiGGST22,Bei2017WorstCaseMD,AmanatidisBM16,AmanatidisBM16a,AmanatidisKS19,JalalyT18,LeonardiMSZ17,GravinJLZ20,Dobzinski2011MechanismsFC,ChenGL11} improving approximation ratios for budget feasible mechanisms under additive, submodular, XOS, or subadditive valuations.
Variants of the problem were also considered, such as incorporating additional downward monotone set family constraints~\cite{AmanatidisBM16a,LeonardiMSZ17,NeogiPS24,HuangHCT23}, considering large market assumptions~\cite{AnariGN14}, considering multi-unit generalizations~\cite{Chan2014TruthfulMP,Amanatidis2023PartialAI,Wu2019MultiunitBF,Klumper2022BudgetFM,Qiao2020TruthfulMD}, multidimensional generalizations~\cite{NeogiPS_multi}, or considering more general budget constraints~\cite{NeogiPS24}.
Variants incorporating experimental design problems were also looked at~\cite{HorelIM14}.

For additive valuations, \cite{GravinJLZ20} achieved an optimal approximation ratio of 2.
For submodular valuations, \cite{BalkanskiGGST22} provide a deterministic budget feasible
mechanism with an approximation ratio of 4.75, which is the current best known
approximation ratio. 
For XOS and subadditive valuations, it is known that no approximation better than
$\sqrt{n}$ is possible using only value oracle queries~\cite{Singer10}, so work on these
valuation classes has always assumed access to at least demand oracles. 
For XOS valuations, an $O(1)$-approximation was known using demand oracles and
XOS oracles~\cite{Bei2017WorstCaseMD,AmanatidisBM16}, but only recently a polytime
$O(1)$-approximation was obtained using only demand oracles~\cite{NeogiPS24}; the latter
work also achieves the current-best approximation for XOS valuations.
For subadditive valuations, there is an exponential-time $O(1)$-approximate budget
feasible mechanism~\cite{NeogiPS24}, and also polytime 
mechanisms that achieve an $O(\log n / \log \log n)$-approximation~\cite{Bei2017WorstCaseMD,BalkanskiGGST22}. 
A polynomial-time $O(1)$-approximation for subadditive functions using demand oracles
still eludes our current understanding. Our work makes progress towards this goal.  

The use of distributions with bounded marginals, as in $\LPmax$, %is not new and
has been considered before in the context of subadditive
valuations~\cite{Dtting2020AnOL,Bangachev2023qPartitioningVE}.
Importantly, \cite{Dtting2020AnOL} use this LP to construct an $O(\log \log n)$-prophet
inequality for online subadditive combinatorial auctions.
The existence of a Nash equilibrium for our item-bidding zero-sum game that achieves
payoff $\frac{1}{2}v(U)$ was implicit in the work of \cite{Bei2017WorstCaseMD,NeogiPS24}.
Variants of our item-bidding game, in the context of first-price and second-price auctions over subadditive valuations, have been analyzed in \cite{Bhawalkar2011WelfareGF,Hassidim2011NonpriceEI,Feldman2015SimultaneousAW}.
Our item-bidding game can also be seen as a generalization of the Colonel Blotto game to subadditive valuations~\cite{borel1953theory}. The Colonel Blotto game has attracted recent attention in theoretical computer science~\cite{Roberson2006TheCB,Vu2018EfficientCO,BoixAdser2020TheMC,Behnezhad2022FastAS}.

\subsection*{Future directions}
Finding a polynomial time constant-factor approximation for subadditive budget feasible
mechanisms remains an interesting open question.
As our proof shows, to do so, it would suffice to find a value of $\kappa $ such that
$\Vmax(\kappa ,U)-\Vmax(\kappa ^2,U) \ge \frac{1}{O(1)}v(U)$.  (Recall that
$\Vmax(\kappa,U)$ is the 
maximum expected value of any distribution over subsets of $U$ with marginals
at most $\kappa $.) Finding such a $\kappa$ would also imply a constant factor prophet
inequality for subadditive combinatorial auctions via static posted prices by the work of
\cite{Dtting2020AnOL}. 
However, \cite{Dtting2020AnOL} provide a counter-example showing that there is a
subadditive valuation $v$ for which no such $\kp$ exists, and that their guarantee of
$\Vmax(\kp,U)-\Vmax(\kp^2,U)\geq v(U)/O(\log\log n)$ is in fact tight.
Instead, they suggest considering \emph{non-uniform} marginal bounds, where each element
$e \in U$ has a separate value $\kappa _e$, and we optimize over distributions over
subsets $S \subseteq U$ satisfying $\Pr[e \in S] \le \kappa _e$. If we are able to find a
non-uniform $\kappa \in (0,1)^U$ vector for which $\Vmax(\kappa ,U)-\Vmax(\{\kappa^2_e\}_e,U)\ge
\frac{1}{O(1)}v(U)$, this would lead to an $O(1)$-approximation for both subadditive
budget-feasible mechanisms and online subadditive combinatorial auctions. 

The key role played by bounded-marginal distributions with large expected value, both 
in the design of budget-feasible mechanisms for subadditive valuations, and
in the design of 
a prophet-inequality for subadditive combinatorial auctions in~\cite{Dtting2020AnOL},  
suggests two pertinent lines of inquiry.
First, it brings into further focus the need for understanding the power of this tool,
in particular, how one may analyze and exploit distributions with non-uniform marginals,
as discussed above.
Second, it raises the question of whether there is a deeper connection between the two,
seemingly disparate problem domains of budget-feasible mechanism design and prophet
inequalities for combinatorial auctions.
In particular, is there a connection between the structural properties of subadditive
functions given by Theorems~\ref{thm:existence_of_dist} and~\ref{thm:good_dist_informal},
and the existence of good distributions with non-uniform marginals?

\section{Preliminaries} \label{prelim}

Let $U$ be the set of elements, and let $n:=|U|$. Let $v:2^U \to \mathbb{R}_+^U$ be the
valuation function. In budget-feasible mechanism design, the valuation function is always
assumed to be normalized, that is, $v(\emptyset)=0$. A valuation function $v:2^U\to
\mathbb{R}_+$ is said to be \emph{subadditive} if $v(S)+v(T)\ge v(S\cup T)$ holds for all
$S,T\subseteq U$. We will also assume that the valuation function is monotone, that is,
$v(S)\le v(T)$ for all $S \subseteq T \subseteq U$. This assumption is without loss of
generality as one can obtain a mechanism for non-monotone subadditive valuations if one
has a mechanism for monotone subadditive valuations with the same approximation ratio (see
\cite{Bei2017WorstCaseMD}). 

For an element $e \in U$ and subset $S \subseteq U$, we use $v(e)$ as a shorthand for
$v(\{e\} )$, and $S+e$, $S-e$ as shorthands for $S \cup \{e\}$ and $S \setminus  \{e\}$
respectively. 
For a set $S \subseteq U$ and a vector $p \in \mathbb{R}^U_+$, we use $p(S)$ as a shorthand to denote $\sum_{e \in S}p_e$.
For $B \in \mathbb{R}_+$ and $S \subseteq U$, let $\Delta^S(B) = \{x \in \mathbb{R}^S_+ : x(S) \le B\}$.
We use $\Delta (B)$ as a shorthand for $\Delta ^U(B)$.
For a cost vector $c \in \mathbb{R}^U$, we denote by $c_{-e}$ the vector in
$\mathbb{R}^{U-e}$ obtained by dropping the coordinate corresponding to $e$. 
Moreover, for $c'_e \in \mathbb{R}$, we denote by $(c'_e,c_{-e})$ the vector $x$ in
$\mathbb{R}^U$ where $x_{e'} = c_{e'}$ for all $e' \in U-e$, and $x_e=c'_e$. 

Throughout, $\OPT := \max\{v(S):c(S)\le B\}$ denotes the optimal value. 
We may assume, using standard preprocessing, that every element $e \in U$ satisfies
$c_e\le B$, as discarding elements that do not satisfy this affects neither the
approximation factor nor truthfulness.
Let $e^*$ denote $\argmax_{e \in U}v(e)$.
Note that the preprocessing guarantees that $\OPT\ge v(e^*)$.

Our mechanism will use a random-partitioning step to compute a good estimate for $\OPT$.
Let $U_1,U_2$ be a random partition of $U$ obtained by placing each element of $U$ independently with probability $\frac{1}{2}$ in $U_1$ or $U_2$.
Let $V_1^* = \max\{v(S):c(S)\le B,S \subseteq U_1\}$ be the optimum over $U_1$, and let $V_2^* = \max\{v(S):c(S)\le B,S \subseteq U_2\}$ be the optimum over $U_2$.
The subadditivity of $v$ can be used to show that both $V_1^*,V_2^*$ are large with high probability.

\begin{lemma}[\cite{NeogiPS24}] \label{lem:partitioning}
	Let $v:2^U\to \mathbb{R}_+$ be a subadditive valuation function.
	Then, we have \linebreak
        \mbox{$\Pr[V_2^* \ge \frac{\OPT}{2},\, V_2^* \ge V_1^* \ge \frac{\OPT-v(e^*)}{4}]\ge \frac{1}{4}$}.
\end{lemma}

\paragraph{Oracles for accessing the valuation.}
Explicitly representing a subadditive valuation function would require space exponential
in $n$, so in order to meaningfully talk about computational efficiency, we assume that the valuation
function is specified via a suitable oracle.
The most natural oracle one could consider is a \emph{value oracle} which takes as input a subset
$S \subseteq U$ and returns its value $v(S)$. 
However, \cite{Singer10} showed that one needs exponentially many value oracle
queries to achieve approximation ratio better than $\sqrt{n}$ for subadditive valuations;
in fact, this lower bound holds even for the subclass of XOS valuations.  
Therefore, in keeping with prior work on budget-feasible mechanism design (when
considering valuations more general than submodular functions), we assume access to a
\emph{demand oracle}, which takes as input prices $p_e$ for each $e \in U$ and returns a
subset $T \in \argmax_{S \subseteq U}\{v(S)-p(S)\}$.  
Demand oracles are often used in the field of mechanism design, and
they have a simple economic interpretation: the oracle returns the subset of items in $U$
with maximum utility, when the cost of purchasing an item $e$ is $p_e$. 
A demand oracle can also be used to compute an approximation to $\OPT$.

\begin{lemma}[\cite{Badanidiyuru2018OptimizationWD}]\label{lem:dem_oracle_approx}
	Let $v:2^U \to \mathbb{R}_+$ be subadditive. For any $\epsilon >0$, one can compute a $(2+\epsilon )$-approximate solution to $\OPT$ in polynomial time using demand oracles.
\end{lemma}

It is important that the demand oracle uses a consistent tie-breaking rule in case there
are multiple sets that attain the maximum $\max_{S \subseteq U}\bigl(v(S)-p(S)\bigr)$.
This can always be achieved, since fixing some ordering of the elements of $U$,
one can always perturb the prices $p$ to obtain a set that is lexicographically
smallest (according to the ordering) among all sets that attain the maximum 
(see \cite{NeogiPS24}). 

\paragraph{Mechanism design.}
A mechanism consists of an allocation rule $f$ along with a payment scheme $p$. 
In the setting of budget-feasible mechanism design, we are given as input the publicly-known information of the valuation function $v:2^U \to \mathbb{R}_+$, the budget $B \in \mathbb{R}_+$, and the seller's reported costs $\{c_e\}_{e \in U}$.
The allocation rule, or algorithm, $f$ outputs a subset $S \subseteq U$ of sellers whose
items are purchased by the mechanism designer. 
We refer to the set $S$ returned by $f$ as the set of \emph{winners}.
The payment scheme $p$ then assigns to each seller a payment $p_e$.
The utility of seller $e$, when her true private cost is $\overline{c}_e\ge 0$ is $u_e = p_e - \overline{c}_e$ if $e \in S$ and $u_e=p_e$ if $e \not\in S$.
Each seller reports a cost that maximizes her own utility.
Note that $f,\{p_e\}_{e \in U}$ are functions of $v,B$ and $c$.
For ease of notation, we treat $v$ and $B$ as fixed and treat $f,\{p_e\}_{e \in U}$ as functions of $c$ only. 
Additionally, $u_e$ is a function of $v,B$, the reported costs $c$ and the true cost $\overline{c}_e$ of seller $e$.
For ease of notation, we treat $v,B$ as fixed and refer to $u_e(\overline{c}_e;c_e,c_{-e})$ as the utility of seller $e$ when the reported costs are $c$ and the true cost of seller $e$ is $\overline{c}_e$.
We seek a mechanism $\mathcal{M}=(f,p)$ satisfying the following properties.

\begin{itemize}
	\item \textbf{Truthfulness.} Each seller $e$ maximizes her utility by reporting her true private cost. That is, for every $\overline{c}_e,c_e,c_{-e}$, we have $u_e(\overline{c}_e;\overline{c}_e,c_{-e}) \ge u_e(\overline{c}_e;c_e,c_{-e})$.
	\item \textbf{Individual Rationality.} $u_e(\overline{c}_e;\overline{c}_e,c_{-e})\ge 0$ for every $e$ and every $\overline{c}_e,c_{-e}$. Note that this implies that $p_e(c)\ge 0$ for all $c$.
	\item \textbf{No positive transfers.} We do not pay sellers whose items we do not purchase, that is, $p_e(c)=0$ whenever $e \not\in f(c)$.
	\item \textbf{Budget Feasibility.} We have $\sum_{e \in U}p_e(c)\le B$ for all $c$. Assuming no positive transfers, this is equivalent to $\sum_{e \in f(c)}p_e(c)\le B$. Note that if $\mathcal{M}$ is individually rational, this implies $c(f(c))\le B$.
\end{itemize}

A \emph{budget feasible mechanism} is a mechanism that satisfies the above properties.
Our mechanism is randomized, and will satisfy the above properties 
with probability $1$, that is, under all realizations of its random bits. A randomized
mechanism that is truthful with probability $1$ is also called a 
{\em universally-truthful} mechanism. 

In addition, we want the mechanism to return a solution that has value close to
$\OPT:=\max\{v(S):c(S)\le B,S\subseteq U\}$. 
We say that $\mathcal{M}$ achieves an $\alpha$-approximation if $f$ is a $\alpha
$-approximate algorithm, that is, $v(f(c))\ge \frac{\OPT}{\alpha }$, where $\alpha \ge 1$.
When $\Mc$ is randomized, we say that $\Mc$ achieves an $\al$-approximation if the expected
value obtained is at least $\OPT/\al$.

Budget-feasible mechanism design is a single-parameter setting. In single-parameter
settings, Myerson's Lemma gives a powerful characterization truthfulness. 
We say that a deterministic algorithm $f$ is \emph{monotone} if for every seller $e \in
U$, and every $c_e\ge c_e' \in \mathbb{R}_+$, and every $c_{-e} \in \mathbb{R}_+^{U-e}$,
if $e \in f(c_e,c_{-e})$ then $e \in f(c'_e,c_{-e})$. That is, $f$ is monotone if a winner
remains a winner upon decreasing its reported cost. 

\begin{theorem}[Truthfulness in single-parameter domains~\cite{Myerson81}]\label{lem:myersons}
Given an algorithm $f$ for budget-feasible mechanism design (i.e., for approximately computing
$\OPT$), there exists payment functions $\{p_e\}_{e \in U}$ such that $(f,p)$ is a
truthful mechanism if and only if $f$ is monotone. 

Moreover, suppose that $f$ is monotone, and $\tau _e = \tau _e(c_{-e}) := \sup\{c_e\ge 0:
e \in f(c_e,c_{-e})\}$ is finite for every $e \in U$ and $c_{-e} \in
\mathbb{R}_+^{U-e}$. Then setting $p_e(c) = \tau _e(c_{-e})$ if $e$ is a winner and 0
otherwise, is a payment scheme that yields a truthful, individually-rational mechanism
with no positive transfers. 
We refer to $\tau_e$ as the \emph{threshold price} or \emph{threshold payment} for seller
$e$.  
\end{theorem}

\paragraph{Optimizing over distributions with bounded marginals.}
Any feasible solution $x$ to $\LPmax$ can be viewed as a distribution over subsets $T$ of
$S$. We abuse notation and denote by $T \sim x$ as a random set $T$ which takes value $T'$
with probability $x_{T'}$, and $\emptyset$ with probability $1-\sum_{T' \subset S}x_{T'}$
(which is non-negative due to the second constraint). 
We use the notation $\Vmax(\kappa ,S)$ to denote the optimal value of $\LPmax$.
The dual of $\LPmax$ is
\begin{alignat}{2}
\text{minimize} & \qquad &\kappa p(S) &+ \mu \tag*{\dlp} \label{distrdual} \\
\text{subject to} & \qquad & p(T) + \mu &\ge v(T) \qquad \forall T \subseteq S \notag \\
&& p,\mu  & \ge 0. \notag
\end{alignat}

$\LPmax$ and its dual can be solved using a demand oracle.
This can be done by the Ellipsoid method, where a demand oracle can be used to construct a separation oracle for the dual of $\LPmax$.
This is done in the following manner:
Given a vector $(p,\mu )$, in order to check if its feasible, we need to check whether $v(T)-p(T) \le \mu $ holds for all $T \subseteq S$.
This can be done by calling a demand oracle on prices $p$, and checking if $\max\{v(T)-p(T):T \subseteq U\} \le \mu $.
If it does, then $(p,\mu )$ is feasible, otherwise the maximum demand set $T^*$ satisfies $v(T^*)-p(T^*)> \mu $ which serves as the violated constraint.
Given this separation oracle for the dual, the Ellipsoid method can then be used to find an optimal solution for both the LP and its dual.

Just as in the case for demand oracles, we will assume that our LP is solved with a consistent tie-breaking rule that selects a canonical optimal solution in the case that there are multiple optimal solutions.
That is, the solver returns a fixed optimal primal solution $x$ and optimal dual solution $(p,\mu )$ once $S$ and $\kappa $ are fixed.

\section{Finding a good distribution over threshold vectors}\label{sec:good_dist}

The main result of this section is to prove \cref{thm:good_dist_informal}, which is restated more precisely below.

\begin{theorem}\label{thm:good_dist}
  Let $B \in \mathbb{R}_+$ and $S^* \subseteq U$. Let $n:=|U|\geq 8$.
  Using demand oracles, one can compute in polynomial time a distribution
  $\mathcal{D}[S^*]$ over vectors $d \in \Delta^{S^*}(B)$ such that  
\[
\Ex_{d\sim \mathcal{D}[S^*]} [ v(\{e \in S^*:c_e\le d_e\} ) ] \ge \frac{v(S^*)}{16 \log\log n}
\qquad \text{for all vectors $c \in \Delta ^{S^*}(B/16 \log \log n)$}.
\]
\end{theorem}

We use the notation $\mathcal{D}[S^*]$ to emphasize that the distribution depends on the set $S^* \subseteq U$.
We begin by proving the following lemma, which is a formal statement of \cref{lem:thresholds_informal}. 
The lemma shows that the dual variables of $\LPmax$ can be used to find a pure strategy $d$ for the $d$-player that achieves a payoff of at least $v(S)-\Vmax(\kappa ,S)$ against all vectors $c \in \mathbb{R}^S_+$ satisfying $c(S)\le \kappa B$.
Fix $S \subseteq U$ and $\kappa \in (0,1)$. 
Fix an optimal solution $(p,\mu )$ to $\dlp$.
We define $d^{S,\kappa }$ to be the vector in $\mathbb{R}^U_+$ given by 
$d^{S,\kappa }_e = p_e\cdot \frac{B}{p(S)}$ for $e \in S$ and $d^{S,\kappa }_e=0$ for $e\in U\setminus S$, 

\begin{lemma}\label{lem:thresholds_simple}
	Let $B \in \mathbb{R}_+$, $\kappa \in (0,1)$, and $S \subseteq U$.
	Then the vector $d:=d^{S,\kappa }$ satisfies the following.
	\begin{itemize}
		\item $d(S) = B$.
		\item For all $c \in \mathbb{R}^S_+$ with $c(S)\le \kappa B$, we have $v(\{e \in S:c_e\le d_e\} ) \ge v(S)-\Vmax(\kappa ,S)$.
	\end{itemize}
\end{lemma}
\begin{proof}
	Let $(p,\mu )$ be the optimal solution to $\dlp$ used to obtain $d$.
	The first property holds by construction: $d(S) = \sum_{e \in S}p_e\cdot \frac{B}{p(S)} = B$.
	Now fix $c \in \mathbb{R}^S_+$ such that $c(S) \le \kappa B$.
	Let $T = \{e \in S: c_e > d_e\}$. Note that $v(\{e \in S:c_e\le d_e\} )=v(S\setminus  T)$. We will show that $v(S\setminus T) \ge v(S)- \Vmax(\kappa ,S)$.
	We have that
	\begin{equation*}
          %\begin{split}
		v(T) \le p(T) + \mu 
		= d(T) \cdot \frac{p(S)}{B} + \mu \le c(T) \cdot \frac{p(S)}{B}+ \mu 
		\le \kappa p(S) + \mu = \Vmax(\kappa ,S)
	  %\end{split}
        \end{equation*}
	where the first inequality is because $d_e<c_e$ for all $e\in T$, and the
        second inequality holds since $c(T) \le \kappa B$. 
	Thus,
	%\begin{align*}
	$v(S\setminus T) \ge v(S) - v(T) \ge v(S) - \Vmax(\kappa ,S)$. %& \qedhere
	%\end{align*}
\end{proof}

Note that, given $S$ and $\kappa $, one can compute the vector $d^{S,\kappa }$ in polynomial time using demand oracles since the optimal solution to the dual of $\LPmax$ can be computed in polynomial time using demand oracles.

We will construct a strategy for the $d$-player that provably obtains an expected payoff of $\Vmax(\kappa ,S^*)-\Vmax(\kappa ^2,S^*)$, for some $\kappa \in (0,1)$. 
It turns out that there always exists $\kappa \in (0,1)$, computable in polynomial time, such that $\Vmax(\kappa ,S^*) - \Vmax(\kappa ^2,S^*) \ge \frac{1}{8\log \log n}v(S^*)$.
This was shown by \cite{Dtting2020AnOL}; to keep exposition self-contained, we include a
proof at the end of this section.
This is the value of $\kappa $ that we will use when solving the LP.

\begin{lemma}[\cite{Dtting2020AnOL}]\label{lem:Dutting}
	Let $S^* \subseteq U$.
	Suppose $n:=|U|\ge 8$.
	There exists $\kappa \in (0,1)$ such that $\Vmax(\kappa ,S^*) - \Vmax(\kappa ^2,S^*) \ge \frac{1}{8\log \log n}v(S^*)$. Moreover, this value of $\kappa $ can be found in polynomial time using demand oracles.
\end{lemma}

Now we are ready to prove \cref{thm:good_dist}.

\begin{proof}[Proof of \cref{thm:good_dist}]
	Let $\kappa \in (0,1)$ be such that $\Vmax(\kappa ,S^*) - \Vmax(\kappa ^2,S^*) \ge
        \frac{v(S^*)}{8\log \log n}$ as guaranteed by \cref{lem:Dutting}.
	The distribution $\mathcal{D}[S^*]$ is obtained by sampling a set $S\sim x$ where $x$ is the optimal solution for $\LPmax[{\kappa,S^*}]$ and then taking the vector $d^{S,\kappa }$.
	By \cref{lem:thresholds_simple}, we have that $\mathcal{D}[S^*]$ is a distribution of vectors in $\Delta^{S^*}(B)$.
	Additionally, note that $x$ can be computed in polynomial time, and also has polynomial sized support, since its the output of the Ellipsoid algorithm.
	Moreover, given $S$ and $\kappa $, the vector $d^{S,\kappa }$ can also be computed in polynomial using demand oracles.
	It follows that a complete description of $\mathcal{D}[S^*]$ (that is, a list of vectors in the support of $\mathcal{D}[S^*]$ and their probabilities) can be computed in polynomial time.

	Now we will show that sampling a vector from $\mathcal{D}[S^*]$ provides a good payoff for the $d$-player.
	We will first show that $\Ex_S[v(S) - \Vmax(\kappa ,S)]\ge \frac{v(S^*)}{O(\log \log n)}$.
	Note that $\Ex_S[v(S)] = \Vmax(\kappa ,S^*)$.
	We claim that $\Ex_S[\Vmax(\kappa ,S)] \le  \Vmax(\kappa ^2,S^*)$.
	Consider the distribution over subsets $T$ of $S^*$, where we first sample a set $S$ from $(x_S)_{S \subseteq S^*}$ where $x$ is the optimal solution to $\LPmax[{\kappa ,S^*}]$ and then sample a set $T$ from $(y_T)_{T \subseteq S}$ where $y$ is the optimal solution to $\LPmax$.
	For any element $e \in S^*$, we have that $\Pr[e \in S] \le \kappa $ and $\Pr[e \in T  \mid  e \in S] \le \kappa $.
	Thus $\Pr[e \in T]\le \kappa ^2$.
	Hence the distribution over $T$ has marginals at most $\kappa ^2$, implying that $\Ex_{T}[v(T)] \le \Vmax (\kappa ^2,S^*)$.
	Finally, observe that $\Ex_{T}[v(T)] = \Ex_{S}[\Vmax(\kappa ,S)]$ by construction.
	Thus, we have that $\Ex_S[\Vmax(\kappa ,S)] \le  \Vmax(\kappa ^2,S^*)$.
	This shows that $\Ex_S[v(S) - \Vmax(\kappa ,S)] \ge \Vmax(\kappa ,S^*) - \Vmax(\kappa ^2,S^*) \ge \frac{v(S^*)}{8\log \log n}$, by our choice of $\kappa $.

	Now, we are almost done. We would like to invoke \cref{lem:thresholds_simple} to show that $\Ex[v(\{e \in S^* : c_e \le d_e\} )] \ge \Ex[v(S)-\Vmax(\kappa ,S)] \ge \frac{v(S^*)}{8\log \log n}$. However, this only works when $c(S) \le \kappa B$. We need to show that the cost vector satisfies this with high probability.

	Fix $c \in \mathbb{R}^{S^*}_+$ satisfying $c(S^*)\le \frac{B}{16 \log \log n}$.
	Recall that $S$ is sampled from the distribution $(x_S)_{S \subseteq S^*}$.
	Since this distribution has marginals at most $\kappa $, we have that $\Ex[c(S)] \le \kappa c(S^*) = \frac{\kappa B}{16 \log \log n}$.
	By Markov's inequality, with probability at least $1-\frac{1}{16 \log \log n}$, we have that $c(S) \le  (16\log \log n)\Ex[c(S)] \le \kappa B$.
	Let $\Omega $ be the event that $c(S) \le \kappa B$, and let $\overline{\Omega  }$ be the event that $c(S)>\kappa B$.
	We have,
	\begin{align*}
		\frac{1}{8\log \log n}v(S^*)&\le \Ex[v(S) - \Vmax(\kappa ,S)]\\
		&= \Ex[v(S) -\Vmax(\kappa ,S)  \mid \Omega ] \Pr[\Omega ] +  \Ex[v(S) - \Vmax(\kappa ,S)  \mid \overline{\Omega } ] \Pr[\overline{\Omega  } ] \\
		&\le \Ex[v(S) - \Vmax(\kappa ,S)  \mid \Omega ]\Pr[\Omega ] + \frac{1}{16 \log \log n}v(S^*)
	\end{align*}
	where we used that fact that $\Pr[\overline{\Omega  }]\le \frac{1}{16 \log \log n}$ and $\Ex[v(S)-\Vmax(\kappa ,S) \mid \overline{\Omega  }]\le v(S^*)$.
	This implies that
	\[\Ex[v(S)-\Vmax(\kappa ,S) \mid \Omega  ]\Pr[\Omega ] \ge \frac{1}{16\log \log n}v(S^*).\]
	Let $Q = \{e \in S : c_e \le d_e\}$. Under the event $\Omega   $, by \cref{lem:thresholds_simple}, we have that $v(Q) \ge v(S) - \Vmax(\kappa ,S)$.
	Thus, \[\Ex_S[v(Q)  \mid \Omega  ]\Pr[\Omega ] \ge \Ex[v(S) - \Vmax(\kappa ,S)  \mid \Omega   ]\Pr[\Omega ] \ge \frac{v(S^*)}{16 \log \log n}.\]
	Hence, the payoff of the $d$-player, which is $v(Q)$, has expected value at least $\frac{v(S^*)}{16 \log \log n}$.
	This concludes the proof.
\end{proof}

\begin{proofof}{Lemma~\ref{lem:Dutting}}
	Let $Z = \{2^{-2^i} : 1\le i\le \log \log n\}$. We will find $\kappa \in Z$ satisfying $\Vmax(\kappa ,S^*) - \Vmax(\kappa ^2,S^*) \ge \frac{1}{O(\log \log n)}v(S^*)$.
	By a telescoping argument, we have,
	\begin{align*}
		\sum_{z \in Z}\Bigl(\Vmax(z,S^*) - \Vmax(z^2,S^*)\Bigr) = \Vmax\Bigl(\tfrac{1}{4},S^*\Bigr) - \Vmax\Bigl(\tfrac{1}{n^2},S^*\Bigr) 
	\end{align*}
	Setting $x_{S^*} = \frac{1}{4}$ and $x_S=0$ for all other sets $S\subset S^*$
        gives rise to a feasible solution for $\LPmax[{\frac{1}{4},S^*}]$ with objective
        $\frac{1}{4}v(S^*)$. Hence $\Vmax\left(\frac{1}{4},S^*\right) \ge
        \frac{1}{4}v(S^*)$. 
	Setting $p_e = v(S^*)$ for all $e \in S^*$ and $\mu = 0$ gives a feasible solution
        to $\dlp[{\frac{1}{n^2},S^*}]$ with objective value $\frac{1}{n^2}p(S^*)+
        \mu=\frac{1}{n}v(S^*)$. This shows that 
        $\Vmax\left(\frac{1}{n^2},S^*\right) \le \frac{1}{n}v(S^*)$. 
	Hence we have 
	\begin{align*}
		\sum_{z \in Z}\Bigl(\Vmax(z,S^*) - \Vmax(z^2,S^*)\Bigr) \ge \left(\frac{1}{4}-\frac{1}{n}\right)v(S^*) \ge \frac{1}{8}v(S^*)
	\end{align*}
	for $n\ge 8$.
	Since $|Z|=\log \log n$, there exists some $\kappa  \in Z$ for which $\Vmax(\kappa ,S^*) - \Vmax(\kappa ^2,S^*) \ge \frac{1}{8\log \log n}v(S^*)$.

	To find $\kappa $ in polynomial time using demand oracles, we may simply compute $\Vmax(z,S^*)-\Vmax(z^2,S^*)$ for each value of $z \in Z$ and take the maximum.
	This can be done in polynomial time as the LP can be solved in polynomial time using demand oracles.
\end{proofof}

\section{The budget feasible mechanism} \label{bfmech}

We now use the result in \cref{sec:good_dist} to design a randomized budget feasible mechanism for
subadditive valuations that achieves an approximation factor of $O(\log \log n)$, 
thereby proving Theorem~\ref{mainthm}.
As mentioned in \cref{sec:tech_overview}, we first randomly partition the set of sellers
into two groups $U_1,U_2$, and get an estimate of $\OPT$ from the sellers in $U_1$ and
discard them. 
We then work over the sellers in $U_2$. 
We compute a demand set over $U_2$, with appropriately chosen prices.
This results in a set $S^* \subseteq U_2$ that has good value compared to $\OPT$.
When $c(S^*) \le \frac{B}{O(\log \log n)}$, we can utilize \cref{thm:good_dist} to prove a
good approximation ratio by sampling a vector $d \sim \mathcal{D}[S^*]$ and return $R =
\{e \in S^*:c_e\le d_e\}$. 
However, we also need to deal with the case when $c(S^*)> \frac{B}{O(\log \log n)}$.
In this scenario, we will show that one can construct a single vector $d'$ that provides
suitable thresholds. 
We show that a suitable choice for $d'$ is to set 
$d'_e$ to be a suitable scaling of $q_e$, where $q\in\mathbb{R}^{S^*}_+$ is an
optimal solution to the LP:
$\max\ q(S^*)$ subject to $q(S) \le v(S^*)-v(S^*\setminus  S)$ for all $S \subseteq S^*$. 
We select the set of elements $A$ for which $c_e \le d'_e$, and \emph{prune} this set to
construct a subset $R' \subseteq A$ satisfying $\Omega (\OPT)-v(e^*)\le v(R') \le
O(\OPT)$.  
We show that constructing $R'$ this way admits budget-feasible threshold payments.
Returning either $R'$ or $e^*$ then gives a good approximation, since if $v(e^*)$ is small
then $v(R')\ge \Omega (\OPT)-v(e^*)$ would be large. 
One caveat here is that we cannot really branch on whether 
$c(S^*)\le \frac{B}{O(\log \log n)}$ as this causes problems with truthfulness. 
So instead, we return one of $R$, $R'$, or $e^*$, each chosen with a suitable
probability. 
The algorithm is described in detail below. 
We show that Algorithm~\ref{bfmech-alg} is a distribution over monotone algorithms, and
so the corresponding payments can be obtained using Theorem~\ref{lem:myersons}. 

\begin{algorithm}
	\caption{Budget feasible mechanism for subadditive valuations}\label{alg:mechanism}\label{bfmech-alg}

	\Input{A valuation function $v:2^U\to \mathbb{R}_+$, budget $B \in \mathbb{R}_+$, and reported costs $\{c_e\}_{e\in U}$}

	\Output{A subset $S \subseteq U$}

        If $n:=|U|<8$, run the exponential-time mechanism from~\cite{NeogiPS24} for
        subadditive valuations. Otherwise, proceed as below.

	Construct a partition $U_1,U_2$ of $U$, where each item $e \in U$ is put into $U_1$ or $U_2$ with probability $\frac{1}{2}$.

	Let $V_1$ be a $(2+\e)$-approximation of $V_1^*$, the optimum value over $U_2$,
        obtained using \cref{lem:dem_oracle_approx}. 
        Let $\al=\frac{1}{4(2+\e)+1}$ and $\beta=1-\al$. \label{v1estim}

	Let $S^* \gets \argmax_{S \subseteq U_2}\left\{v(S)- \frac{V_1}{2B}\cdot c(S)\right\}$. 
        \label{demdset}

	Use \cref{thm:good_dist} to get a distribution $\mathcal{D}[S^*]$ over vectors $d
        \in \Delta^{S^*}(B)$ such that $\Ex_{d \sim \mathcal{D}[S^*]}[v(\{e \in S^* :
          c_e\le d_e\} )] \ge \frac{v(S^*)}{16 \log  \log n}$ for all $c \in \Delta
        ^{S^*}\left(\frac{B}{16 \log \log n}\right)$. \label{gooddist}

	Sample $d \sim \mathcal{D}[S^*]$.

	Let $R \gets \{e \in S^* : c_e \le d_e\}$. \label{rset}

	Compute an optimal solution $q\in\R^{S^*}$ to the following LP: 
        \begin{equation}
          \max \quad  q(S^*) \qquad \text{s.t.} \qquad 
          q(S)\le v(S^*)-v(S^*\setminus S) \quad \forall S \subseteq S^*, \qquad q\geq 0. 
          \tag{P} \label{xoslp}
        \end{equation}
        \label{xosthresh} 

	Let $A \gets \{e \in S^* : c_e \le  q_e\cdot\frac{4B}{V_1}\}$. \label{threshprune}

	Let $R'$ be a maximal prefix of $A$ satisfying $q(R') \le \frac{1}{4}V_1$. \label{rpset}
	
	\Return $R$ with probability $0.8\al$, $R'$ with probability $0.8\beta$, and $e^*$
        with probability $0.2$. 
\end{algorithm}

For step~\ref{xosthresh}, it is important that the vector $q$ is computed with a consistent tie-breaking rule. 
That is, $q$ is fixed once $S^*$ and $v$ are fixed.

We assume that $n\geq 8$, as otherwise, the guarantees follow from~\cite{NeogiPS24}.
Theorem~\ref{mainthm} follows from Lemmas~\ref{polytime}, which shows polynomial running
time, Lemma~\ref{approx}, which proves the approximation guarantee, and
Lemma~\ref{budgetfeas}, which proves truthfulness and budget feasibility.

\begin{lemma} \label{polytime}
	\cref{alg:mechanism} can be implemented in polynomial time using at most polynomial many demand oracle queries.
\end{lemma}
\begin{proof}
	Step~\ref{v1estim} and step~\ref{gooddist} can be implemented using demand oracles by \cref{lem:dem_oracle_approx} and \cref{thm:good_dist} respectively.
	We show how to implement step~\ref{xosthresh}.
	To solve \eqref{xoslp} in step~\ref{xosthresh}, it suffices to give a separation oracle
        for the constraints of the LP. One can then use the Ellipsoid method to solve it. 
        Let $q \in \mathbb{R}_+^{S^*}$. Observe that
        \begin{equation*}
          \begin{split}
            q(S) \le v(S^*)-v(S^*\setminus  S) \quad \forall S\sse S^* & \iff 
            q(S^*)-q(S^*\sm S)\leq v(S^*)-v(S^*\sm S)\quad \forall S\sse S^* \\
            & \iff v(T)-q(T)\leq v(S^*)-q(S^*)\quad \forall T\sse S^*.
          \end{split}
          \end{equation*}
        The final condition above can be easily checked using a demand-oracle query with
        prices $q$.% 
        \footnote{More precisely, we set the price of an element $e\in S^*$ to be $q_e$,
          and set the prices of elements outside of $S^*$ to be $\infty$, or some very
          large value. Then the demand-oracle query effectively amounts to a demand-oracle query
          over the ground set $S^*$.}
        If the oracle returns some $T\sse S^*$ with $v(T)-q(T)>v(S^*)-q(S^*)$, then the
        LP constraint for $S^*\sm T$ is violated, and otherwise $q$ is feasible.
\end{proof}

Before analyzing the approximation ratio, we prove a lemma regarding the cost of the demand set.

\begin{lemma}\label{lem:dem_bounds}
	For every $S \subseteq S^*$, we have that $ \frac{V_1}{2B} c(S)\le v(S^*)-v(S^*\setminus S) \le v(S)$.
\end{lemma}
\begin{proof}
	Let $S \subseteq S^*$. By construction of $S^*$, we have that $v(S^*)-\frac{V_1}{2B} c(S^*) \ge v(S^*\setminus S) - \frac{V_1}{2B} c(S^*\setminus S)$.
	Rearranging, we get that $\frac{V_1}{2B} c(S) = \frac{V_1}{2B} c(S^*) - \frac{V_1}{2B} c(S^*\setminus S) \le v(S^*)-v(S^*\setminus S)$, showing the first inequality.
	The second inequality follows by subadditivity of $v$.
\end{proof}

Now we are ready to prove the approximation ratio.

\begin{lemma} \label{approx}
  \cref{alg:mechanism} returns a solution $T$ such that $\Ex[v(T)] \ge \frac{1}{O(\log \log n)}\cdot\OPT$.
\end{lemma}
\begin{proof}
	Recall that $V_i^* = \max\{v(S) : c(S)\le B, S \subseteq U_i\}$ for each 
        $i \in \{1,2\}$, and $v(e^*)=\max_{e\in U}v(e)$. 
	Let $\Gamma $ be the event that $V_2^*\ge V_1^* \ge \frac{\OPT-v(e^*)}{4}$ and $V_2^* \ge \frac{\OPT}{2}$.
	By \cref{lem:partitioning}, $\Gamma $ holds with probability $\ge \frac{1}{4}$.
	Assume that the event $\Gamma $ happens.
        So we have $V^*_2\geq\frac{\OPT}{2}$ and
        $V_1\geq\frac{V^*_1}{2+\e}\geq\frac{\OPT-v(e^*)}{4(2+\e)}=\frac{\al}{\beta}\cdot\bigl(\OPT-v(e^*)\bigr)$.

	First, let us show that $v(S^*) \ge \frac{\OPT}{4}$.
	To see this, let $O^*_2 = \argmax\{v(S):c(S)\le B, S\subseteq U_2\}$, so $v(O_2^*)=V_2^*$. 
	So we have 
        \[
        v(S^*)\ge v(S^*)- \frac{V_1}{2B}\cdot c(S^*) \ge v(O^*_2) - \frac{V_1}{2B}\cdot
        c(O^*_2) 
        \ge V_2^* - \frac{V_1}{2} \ge \frac{V_2^*}{2}\geq \frac{\OPT}{4}.
        \]
        If $c(S^*)\le \frac{B}{16 \log \log n}$, then by \cref{thm:good_dist}, it follows
        that $\Ex_{d\sim \mathcal{D}[S^*]}[v(R)] \ge \frac{v(S^*)}{16 \log \log n} \ge\frac{\OPT}{64 \log \log n}$.
        So if $c(S^*)\leq\frac{B}{16\log\log n}$, then the expected value obtained is at
        least 
        \[
        \Pr[\Gm]\cdot 0.8\al\cdot\frac{\OPT}{64\log\log n}\geq\frac{\OPT}{(2880+1280\e)\log\log n}.
        \]

	Now suppose $c(S^*)>\frac{B}{16 \log \log n}$. 
	In this case, we argue that $R'$ has good value.
	By \cref{lem:dem_bounds}, the vector $\frac{V_1}{2B}c$ is feasible for
        \eqref{xoslp}.
	Hence the optimal solution $q$ to \eqref{xoslp} satisfies 
        $q(S^*)\ge  \frac{V_1}{2B}\cdot c(S^*)$.
	We have 
	\begin{align*}
	  q(A) = q(S^*)-q(S^*\setminus  A) &
          \ge \frac{V_1}{2B}\cdot c(S^*)-\frac{V_1}{4 B}\cdot c(S^*\setminus A) 
          \ge \frac{V_1}{4B}\cdot c(S^*) > \frac{V_1}{64 \log \log n}
	\end{align*}
	Thus, $v(A)\ge v(S^*)-v(S^*\sm A)\geq q(A)\ge \frac{V_1}{64 \log \log n}$. 
	This implies that 
        \begin{equation*}
          \begin{split}
            v(R') & \ge \min\Bigl\{\tfrac{V_1}{4}-v(e^*),\,\tfrac{V_1}{64\log\log n}\Bigr\}
            \geq \min\Bigl\{\tfrac{\al}{\beta}\cdot\bigl(\OPT-v(e^*)\bigr)-v(e^*),\,
            \tfrac{\al}{\beta}\cdot\tfrac{\OPT-v(e^*)}{64\log\log n}\Bigr\} \\
            & = \frac{\al}{\beta}\cdot\min\Bigl\{\OPT-\tfrac{v(e^*)}{\al},\,\tfrac{\OPT-v(e^*)}{64\log\log n}\Bigr\}
          \end{split}
        \end{equation*}
        by the construction of $R'$.
        So if $c(S^*)>\frac{B}{16\log\log n}$, the expected value returned is at least
        \begin{equation*}
          \begin{split}
            0.2\cdot v(e^*)+\frac{1}{4}\cdot 0.8\beta\cdot\tfrac{\al}{\beta}\cdot\min
            \Bigl\{\OPT-& \tfrac{v(e^*)}{\al},\,\tfrac{\OPT-v(e^*)}{64\log\log n}\Bigr\} 
            \geq \min\Bigl\{0.2\al\cdot\OPT,\, 0.2\al\cdot\tfrac{\OPT}{64\log\log n}\Bigr\}
            \\ & = 0.2\al\cdot\frac{\OPT}{64\log\log n}=\frac{\OPT}{(2880+1280\e)\log\log n}. 
          \end{split}
        \end{equation*}
        So in all cases, the expected value returned is $\OPT/O(\log\log n)$.
\end{proof}

Finally, we prove truthfulness and budget feasibility.

\begin{lemma} \label{budgetfeas}
\cref{alg:mechanism} can be combined with suitable payments to obtain a mechanism that
is truthful, individually rational, has no positive transfers, and is budget feasible,
with probability $1$. 
\end{lemma}

\begin{proof}
We argue that Algorithm~\ref{alg:mechanism} is monotone under all realizations of its
random bits, which immediately implies by \cref{lem:myersons} that using threshold prices
as payments yields truthfulness, individual rationality, and no positive transfers, with
probability $1$. 

\smallskip
\noindent
\textbf{Truthfulness, individual rationality, and no positive transfers.}
	Fix a realization of the random bits, so that the partition $U_1,U_2$ of $U$ is
        fixed, and the set returned (either $R,R'$ or $e^*$) is fixed. 
	Fixing the random bits also fixes the vector $d$ sampled from $\mathcal{D}[S]$
        for any given set $S\sse U$.
	We need to show that if $e \in U$ is a winner, and $e$ decreases their cost from
        $c_e$ to $c'_e < c_e$, then $e$ remains a winner. 

	First, suppose that the set returned is $e^* = \argmax_{e \in U}v(e)$.
	Note that $e^*$ does not depend on the cost of any element in $U$, thus $e^*$
        decreasing their cost does not affect whether $e^*$ remains a winner or not. 

	Suppose that the set returned is $R$.
	Let $e \in R$ with cost $c_e$, and suppose that $e$ decreases their reported cost to $c'_e<c_e$.
	We claim that the set $S^*$ computed in step~\ref{demdset} is the same for both
        the inputs $c:=(c_e,c_{-e})$ and $c':=(c'_e,c_{-e})$.
        Let $\Sc(c)$ denote the set of optimal solutions to
        $\max\,\{v(S)-\frac{V_1}{2B}\cdot c(S): S\sse U_2\}$ and $\Sc(c')$ denote the set
        of optimal solutions to $\max\,\{v(S)-\frac{V_1}{2B}\cdot c'(S): S\sse U_2\}$
        (note that $U_2$ is the same for both $c$ and $c'$ since the random bits are fixed).
        Let $S^*(c)\in\Sc(c)$ and $S^*(c')\in\Sc(c')$ denote the sets
        computed in step~\ref{demdset} for the inputs $c$ and $c'$ respectively.
	We argue that $S^*(c)\in\Sc(c')\sse\Sc(c)$.
	Since the demand oracle uses a consistent tie-breaking rule, this implies that it
        must return the set $S^*(c)$ also under the input $c'$, i.e., $S^*(c)=S^*(c')$.

	For a set $H\subseteq U$ and any $\tc\in \mathbb{R}^U_+$, define 
        $\demd(H,\tc) = v(H)-\frac{V_1}{2B}\cdot\tc(H)$.
        Consider any $H\in\Sc(c')$. 
        We have $\demd(S^*(c),c)\geq\demd(H,c)$ and
        $\demd(H,c')\geq\demd(S^*(c),c')$, by definition. Adding these inequalities and simplifying, we
        obtain that $c_{e'}-c_e\geq c(H)-c'(H)$. This can only happen if $e\in H$, which
        implies that this inequality is in fact an equality, and hence the two
        inequalities that were added to yield this must also be equalities.
        So we have $\demd(S^*(c),c)=\demd(H,c)$ and $\demd(H,c')=\demd(S^*(c),c')$. The
        former implies that $H\in\Sc(c)$, and so $\Sc(c')\sse\Sc(c)$, and the latter
        implies that $S^*(c)\in\Sc(c')$.
	Thus, $S^*(c) = S^*(c')$.

	From this it follows that the distributions $\mathcal{D}[S^*(c)]$ and
        $\mathcal{D}[S^*(c')]$ are the same. 
	Since the random bits used by the mechanism are fixed, this implies that the
        vector $d$ sampled from $\mathcal{D}[S^*]$ is the same under both $c$ and
        $c'$. 
	It follows that $R$ is the same set under both the inputs $c$ and $c'$.
	Hence, since $e \in R$ under input $c$, we also have $e \in R$
        under input $c'$.

	Now, suppose that the set returned is $R'$. Let $e \in R'$ and suppose that $e$
        decreases their reported cost to $c'_e<c_e$. 
	As before, let $c = (c_e,c_{-e})$ and $c'=(c'_e,c_{-e})$.
	We claim that the sets $A$ and $R'$ are the same under both inputs $c$ and $c'$.
	As argued above, the set $S^*$ is the same under both inputs $c$ and $c'$. 
	This implies that the vector $q$ computed in step~\ref{xosthresh} is also the same
        under $c$ and $c'$ (since the computation of $q$ is fixed once $S^*$ and $v$ are fixed). 
        This then also implies that $A$ is the same under both inputs: since $e\in R'\sse A$
        under input $c$, we have $c'_e<c_e \le q_e \cdot \frac{4B}{V_1}$; for any other
        element $f\neq e$ belonging to $A$ under input $c$, we have $c'_f=c_f \le q_f
        \cdot \frac{4B}{V_1}$.  
	Finally, since the computation of $R'$ does not consider the costs, it follows
        that $R'$ is also the same under both reported costs $c$ and $c'$. 
        In particular, since $e \in R'$ under input $c$, we also have $e \in R'$ under input $c'$.

	Hence we have shown that, regardless of the realizations of the random bits, we
        have a monotone algorithm. 

        \medskip\noindent
	\textbf{Budget feasibility.}
	Now we show that the threshold payments yield budget feasibility, regardless of
        the realizations of the random bits used by the algorithm.
	Again, fix a realization of the random bits.
	Let $\tau _e$ be the threshold price for $e \in U$ as defined in \cref{lem:myersons}.
	If the mechanism returns $e^*$, then $\tau _{e^*}=B$ and $\tau _e = 0$ for all
        $e\neq e^*$, so we obtain budget feasibility. 

	Otherwise, the mechanism returns either $R$ or $R'$. Note that we have argued
        above that if a winner $e$ changes her reported cost and remains a winner, then
        the set returned, as also the ``intermediate'' sets $S^*$, $A$, also do not
        change. This property, which is called no-bossiness in~\cite{NeogiPS24}, makes it
        quite convenient to reason about threshold payments. 

	Suppose that the mechanism returns $R$. 
	Fix $e \in U$ that is a winner under reported costs $c$. Let $S^*=S^*(c)$ be the
        set computed in step~\ref{demdset}, and $d$ be the vector obtained by sampling
        from $\mathcal{D}[S^*]$ using the fixed random bits. Since $S^*$ and $d$ do not
        change if $e$ remains a winner, it is easy to see that, by design, step~\ref{rset}
        ensures that $\tau_e\leq d_e$. It follows that the total
        payment is at most $d(S^*)\leq B$, since $d\in\Dt^{S^*}(B)$.  

	Now, suppose that the mechanism returns $R'$, and consider a winner $e\in R'$
        under input $c$. Since the vector $q$ and the sets $A,R$ computed in
        steps~\ref{xosthresh}--\ref{rpset} do not change if $e$ remains a winner,
        by design, step~\ref{threshprune} ensures that 
        $\tau_e\leq q_e\cdot\frac{4B}{V_1}$. So since $q(R')\leq\frac{V_1}{4}$, the total
        payment is at most $B$.

        \medskip
        Note that the above discussion also shows that the threshold for a player $e$ can
        be computed in polytime. The threshold value $\tau'_e$ for $e$ to belong to the
        set $S^*$ can be obtained using a demand oracle. If the set output is $R$, then
        the threshold for $e$ is $\min\{\tau'_e,d_e\}$. If the set output is $R'$, then
        the threshold for $e$ to lie in $A$ is 
        $\tau^A_e=\min\bigl\{\tau'_e,q_e\cdot\frac{4B}{V_1}\bigr\}$. Note that fixing $A$ also
        fixes the set $R'$. Let $\bA$ be the $A$-set obtained under input $(c'_e,c_{-e})$
        when $e$ belongs to it, i.e., when  $c'_e<\tau^A_e$; let $\bR$ be the
        corresponding $R'$-set. Note $\bA$ and $\bR$ depend only on $c_{-e}$. One can
        infer that $\tau_e=\tau^A_e$, if $e\in\bR$ and is $0$ otherwise.
\end{proof}

% Bibliography
\bibliographystyle{alpha}
\bibliography{improvbudgetfeas-refs}

% Appendix
\appendix

\section{Proofs omitted from the main body}

\begin{proof}[Proof of \cref{thm:existence_of_dist}]
	Consider the following LP
	\begin{alignat*}{2}
	  \text{maximize} & \quad & t \qquad \qquad & \tag{P'} \label{primal} \\
          \text{subject to} && \sum_{d \in K} \frac{v(\{e \in U:c_e\le d_e\} )}{v(U)}\cdot
          x_d & \ge t \qquad \forall c \in K \\
	  && \sum_{d \in K}x_d & = 1 \\
	  && x & \ge 0.
        \end{alignat*}
	Let $M \in \mathbb{R}^{K\times K}$ be a matrix with rows and columns indexed by
        elements in $K$ defined as $M_{c,d} = \frac{v(\{e \in U:c_e\le d_e\} )}{v(U)}$ for
        every $c,d \in K$. 
	Let $\vec{1}$ be the all-ones vector and let $J$ be the all-ones matrix of
        dimension $|K| \times |K|$. 
	Then \eqref{primal} can be written succinctly as
        \begin{equation*}
	  \text{maximize} \quad t \qquad \text{subject to} \qquad  Mx \ge  t\vec{1}, \quad
	  \vec{1}^\intercal x = 1, \quad x \ge 0.
	\end{equation*}	
        The dual LP is
        \begin{equation}
	  \text{minimize} \quad \rho \qquad \text{subject to} \qquad  M^\intercal y \le  \rho\vec{1}, \quad
	  \vec{1}^\intercal y = 1, \quad y \ge 0. \tag{D'} \label{dual}
	\end{equation}	
	Let $(y,\rho )$ be a feasible solution to \eqref{dual}. 
	By subadditivity of $v$, for every $c,d \in K$, we have that 
        $v(\{e \in U:c_e\le d_e\} )+v(\{e \in U:d_e\le c_e\} ) \ge v(U)$. 
	Hence, we get that 
        $M_{c,d}+M_{d,c} = \frac{v(\{e \in U:c_e\le d_e\} )+v(\{e \in U:d_e\le c_e\} )}{v(U)} \ge 1$ for every $c,d \in K$. 
	In other words, $M+M^\intercal \ge J$, that is, every entry of $M+M^\intercal $ is
        at least the corresponding entry of $J$. 
	Now, 
	\begin{align*}
		\rho = \rho (y^\intercal \vec{1}) = y^\intercal (\rho \vec{1}) \ge y^\intercal My = \frac{1}{2}y^\intercal (M+M)y = \frac{1}{2}y^\intercal (M+M^\intercal )y \ge \frac{1}{2}y^\intercal J y = \frac{1}{2},
	\end{align*}
	where we used the fact that $\vec{1}^\intercal y=1$ and the fact that $y^\intercal My = y^\intercal M^\intercal y$.
	Hence we conclude that $\rho \ge \frac{1}{2}$ for every feasible solution of \eqref{dual}.
	In particular, this implies that the optimal value of \eqref{dual}, and hence
        \eqref{primal}, is at least $\frac{1}{2}$. 

	Now let $(x,t)$ be an optimal solution for \eqref{primal}.
	Consider the distribution $\mathcal{D}^*$ over vectors in $K$, where the
        probability of sampling a vector $d \in K$ is $x_d$. 
	Then, for all $c \in K$, we have $\Ex_{d\sim \mathcal{D}^*}[v(\{e \in U:c_e\le
          d_e\} )] = \sum_{d \in K} v(\{e \in U:c_e\le d_e\} ) x_d \ge t\cdot v(U) \ge \frac{1}{2}v(U)$.
	This proves the lemma.
\end{proof}

\begin{lemma}\label{lem:no_pure_strategy}
	For every $d \in \mathbb{R}^U_+$ with $d(U)\le B$, there exists $c \in \mathbb{R}^U_+$ with $c(U)\le B$ such that $v(\{e \in U:c_e\le d_e\} ) \le v(e^*)$.
\end{lemma}
\begin{proof}
	Let $Q = \{e \in U:d_e>0\}$. 
	Let $e'$ be an arbitrary element in $Q$.
	Let $\epsilon = \frac{1}{n}\min_{e \in Q}\{d_e\} > 0$.
	Construct the vector $c$ by setting $c_e = d_e+\epsilon $ for all $e \in U-e'$, and set $c_{e'} = d_{e'} - (n-1) \epsilon $.
	Then note that $c(U)=d(U)\le B$. Moreover, $c_e \ge 0$ for all $e \in U$ by our choice of $\epsilon $.
	Now, $v(\{e \in U:c_e\le d_e\} ) = v(e') \le  \max_{e \in U}\{v(e)\} = v(e^*)$.
\end{proof}

\end{document}